\documentclass[10pt]{article}
\usepackage{mathrsfs}
\usepackage{amsfonts}
\usepackage{amsthm,amsmath,amssymb,anysize}
\newtheorem{lemma}{Lemma}[section]
\newtheorem{theorem}[lemma]{Theorem}

\usepackage[numbers,sort&compress]{natbib}
\marginsize{4.2cm}{4.2cm}{3.6cm}{3cm}%
\numberwithin{equation}{section}

\title{\textsf{The minimal number of generators for simple Lie superalgebras }}
\author{\textsc{ Wende Liu$^{1,2}$}\footnote{Supported by the NSF for Distinguished Young Scholars, HLJ Province (JC201004) and  the NSF
  of China (10871057)}\ \textsc{and}
    \textsc{Liming Tang$^{1,2}$}\footnote{Correspondence:  \texttt{wendeliu@ustc.edu.cn}  (W. Liu), \texttt{limingtang@hrbnu.edu.cn} (L. Tang)
 }
  \\
  \small\textit{$^{1}$Department of Mathematics},
  \small\textit{Harbin Institute of Technology}\\
  \small\textit{Harbin 150006, China}\\
 \small\textit{$^{2}$School of Mathematical Sciences},
 \small \textit{Harbin Normal University} \\
 \small \textit{Harbin 150025, China}}
\date{ }
\begin{document}
\maketitle
\begin{quotation}
{\small\noindent \textbf{Abstract}: Using the classification theorem
due to Kac we prove that any finite dimensional simple
 Lie superalgebra over an algebraically closed
field of characteristic 0 is generated by one element.

\vspace{0.05cm} \noindent{\textbf{Keywords}}: Classical Lie
superalgebra; Cartan Lie superalgebra; generator

\vspace{0.05cm} \noindent \textbf{Mathematics Subject Classification 2000}: 17B05, 17B20, 17B70}
\end{quotation}
 \setcounter{section}{-1}
\section{Introduction}
Throughout we work over
  an algebraically closed field  of characteristic zero, $\mathbb{F}$, and all the vector
spaces and algebras are assumed to be finite dimensional.
  Our principal aim is to determine the minimal
number of generators for simple Lie
superalgebras.  We prove that
 any simple Lie superalgebra is generated by one element in the super sense, that is,
  any simple Lie superalgebra coincides with the smallest sub-Lie superalgebra containing
  some fixed element. We know that for finite simple groups only the groups of prime
  orders can be generated by one element and that ???a simple Lie algebra is never generated by one element. Our results are not surprising since in a Lie superalgebra the quare of an odd element is even and not necessarily zero.
  As in finite simple group or simple Lie algebra cases, our proof is  dependent on the classification theorem of simple Lie superalgebras \cite{Kac}, but not a one-by-one checking.

This study is mainly motivated by two papers of Bois mentioned as follows.   In 2009,  Bois \cite{MBJ} proved that any simple Lie algebra in
arbitrary characteristic  $p\neq 2,3$ is generated by  2 elements
and moreover, the classical Lie algebras and the graded Cartan type simple Lie algebras $W(1,\underline{n})$
(Zassenhaus algebras) can be generated by 1.5 elements, that is,
 any given nonzero element can be paired a suitable element such that these two elements
  generate  the whole algebra.  Later, as a continuation of this work, Bois \cite{MBJ1} showed that the simple
graded Lie algebras of Cartan type $W(m,\underline{n})$ with $m\neq 1$ and the ones  of the remaining Cartan types $S, H$ and $K$ are never
generated by 1.5 elements. Papers \cite{MBJ,MBJ1} contain  a
 considerable amount of information in characteristic $p$ and cover the earlier results in characteristic $0$: In 1976,  Ionescu \cite{GIT} proved that a
simple Lie algebra  over the field of complex numbers is
generated by $1.5$ elements; In 1951, Kuranashi \cite{GKM} proved
that a semi-simple Lie algebra in characteristic 0 is generated by 2
elements.

By the classification theorem
  \cite{Kac},  a simple Lie superalgebra
(excluding simple Lie algebras) is  either a classical
Lie superalgebra or a  Cartan Lie superalgebra (see also \cite{MS}).
 The Lie algebra (even part) of a classical Lie superalgebra  is reductive
 and meanwhile there exists a similarity in the
structure side between the Cartan Lie superalgebras in
characteristic 0 and the simple graded Lie algebras of Cartan type in
characteristic $p$. In
view of the observation above, we began this study in 2009. Since the Lie algebra of a classical  Lie superalgebra is reductive and the odd part
decomposes into at most two irreducible components as
adjoint modules (see \cite{Kac} or \cite{MS}), we first proved that the
Lie algebra of a classical Lie superalgebra is generated by two elements
 and then obtained that any classical Lie superalgebra is generated by two
  (non-homogeneous)  elements in non-super sense; As expected, we also proved that any Cartan Lie
   superalgebra is generated by two (non-homogeneous)  elements (see  arXiv 1103. 4242vl math. ph).
   Thus we obtained that any simple Lie superalgebra is generated by two (non-homogeneous)  elements
   in the non-super sense and then we submitted the manuscript to a journal for publication. Later,
   a conversation with Professor Yucai Su during a workshop on Lie theory (organized by Professor Bin Shu, spring of 2011) led  us to  consider  the question:
 Determine all the simple Lie superalgebras which are generated by one element??? (equivalently, by two homogeneous elements in the non-super sense).

Let us briefly explain the outline and ideas in this improved
version. A simple fact  is that
$[L_{\bar{1}},L_{\bar{1}}]=L_{\bar{0}}$ for a simple Lie
superalgebra $L.$ So, for a classical Lie superalgebra $L$, our
discussion is mainly based on the weight decomposition of
$L_{\bar{1}}$ as $L_{\bar{0}}$-module relative to the standard
Cartan subalgebra???: We find the desired generators by starting
from the sum of all the odd weight vectors and use the fact that
$L_{\bar{1}}$ as $L_{\bar{0}}$-module is irreducible or a direct sum
of two irreducible submodules. For a Cartan Lie superalgebra $L$, we
know that $L$ is generated by its local part $L_{-1}+L_{0}+L_{1}$
with respect to the standard grading and moreover, the null $L_{0}$
is a Lie algebra and $L_{i}$ as $L_{0}$-module, $i=\pm 1$, is
irreducible or a direct sum of two irreducible submodules. Then,
considering the weight space decomposition relative to the standard
Cartan subalgebra of $L_{0}$???, we find the desired generators by
choosing a weight vector in each irreducible submodule of $L_{i}$.
   The proofs??? of main conclusions are
constructive and provide an explicit description of the generator candidates. The process involves certain computational techniques and we use certain information about classical Lie superalgebras from \cite{Z}.

  In this paper we write $\langle X\rangle$ for
the sub-Lie superalgebra generated by a subset $X$ in a Lie
superalgebra.
\section{Classical Lie superalgebras}

A  classical Lie superalgebra by definition is  a simple Lie
superalgebra
 $L=L_{\bar{0}}\oplus L_{\bar{1}}$ for which
  $L_{\bar{1}}$ as $L_{\bar{0}}$-module is
completely reducible  \cite{Kac,MS}.  The information of classical
Lie superalgebras is as follows \cite{Kac}:
$$\begin{tabular}{|l|l|l|} \multicolumn{3}{c}{
 Table 1.1}
\\[5pt]\hline
\multicolumn{1}{|c|}{$L$}&\multicolumn{1}{|c|}{$L_{\bar{0}}$}&\multicolumn{1}{|c|}{$L_{\bar{1}}$ as $L_{\bar{0}}$-module}\\
\hline
 ${\rm{A}}(m,n),\; m,n\geq 0, n\neq m$&
${\rm{A}}_{m}\oplus {\rm{A}}_{n}\oplus
\mathbb{F}$&$~~\mathfrak{sl}_{m+1}\otimes\mathfrak{sl}_{n+1}\otimes \mathbb{F}\oplus (\mbox{its dual})$\\
\hline
 ${\rm{A}}(n,n),\;n> 0$&
${\rm{A}}_{n}\oplus {\rm{A}}_{n}$&$~~\mathfrak{sl}_{n+1}\otimes\mathfrak{sl}_{n+1}\oplus (\mbox{its dual})$\\
\hline ${\rm{B}}(m,n),\;m\geq 0,n>0$&${\rm{B}}_{m}\oplus
{\rm{C}}_{n}$&$~~\mathfrak{so}_{2m+1}\otimes
 \mathfrak{sp}_{2n}$\\
 \hline ${\rm{D}}(m,n),\;m\geq 2,n>0$&${\rm{D}}_{m}\oplus
{\rm{C}}_{n}$&~~$\mathfrak{so}_{2m}\otimes
 \mathfrak{sp}_{2n}$\\
  \hline
 ${\rm{C}}(n),\;n\geq 2$&${\rm{C}}_{n-1}\oplus \mathbb{F}$&$~~\mathfrak{csp}_{2n-2}\oplus (\mbox{its dual})$\\
 \hline
  ${\rm{P}}(n),\;n\geq 2$&${\rm{A}}_{n}$&
  $~~\Lambda^{2}\mathfrak{sl}^{*}_{n+1}\oplus{\rm{S}}^{2}\mathfrak{sl}_{n+1}$\\ \hline
 ${\rm{Q}}(n),\;n\geq 2$&
 ${\rm{A}}_{n}$& $~~{\rm{ad}}\mathfrak{sl}_{n+1}$\\ \hline
${\rm{D}}(2,1;\alpha),\;\alpha\in \mathbb{F}\setminus \{-1,0\}$&
 ${\rm{A}}_{1}\oplus
 {\rm{A}}_{1}\oplus {\rm{A}}_{1}$&
 $~~\mathfrak{sl}_{2}\otimes\mathfrak{sl}_{2}\otimes\mathfrak{sl}_{2}$\\ \hline
${\rm{G}}(3)$& $\mathfrak{G}_{2}\oplus
 {\rm{A}}_{1}$&~~$\mathfrak{G}_{2}\otimes
 \mathfrak{sl}_{2}$\\ \hline
 ${\rm{F}}(4)$&
${\rm{B}}_{3}\oplus  {\rm{A}}_{1}$&~~$\mathfrak{spin}_{7}\otimes
\mathfrak{sl}_{2}$\\ \hline
\end{tabular}\\$$

 From Table 1.1, one sees that $L_{\bar{0}}$ is reductive and
 $L_{\bar{1}}$ as
$L_{\bar{0}}$-module is irreducible or a direct sum of
 two irreducible submodules.

 Throughout this section $L$ denotes a classical
Lie superalgebra with the standard Cartan subalgebra $H $. The
corresponding weight (root) decompositions are
\begin{eqnarray*}
&& L_{\bar{0}}=H\oplus
\bigoplus_{\alpha\in \Delta_{\bar{0}}}L_{\bar{0}}^{\alpha},\qquad
L_{\bar{1}}=\bigoplus_{\beta\in
\Delta_{\bar{1}}}L_{\bar{1}}^{\beta};\\
&& L=H\oplus \bigoplus_{\alpha\in
\Delta_{\bar{0}}}L_{\bar{0}}^{\alpha}\oplus\bigoplus_{\beta\in
\Delta_{\bar{1}}}L_{\bar{1}}^{\beta}.
\end{eqnarray*}
Write
   $$
   \Delta:=\Delta_{\bar{0}}\cup \Delta_{\bar{1}}\quad\mbox{and}\quad L^{\gamma}:=L_{\bar{0}}^{\gamma}\oplus L_{\bar{1}}^{\gamma}\quad\mbox{for}\;\gamma\in \Delta.
   $$
   Note that the standard  Cartan subalgebra of a classical Lie superalgebra is diagonal:
\begin{equation}\label{eq-classical-root-vector}
 \mathrm{ad}h(x)=\gamma(h)x \quad \mbox{for}\;h\in H,\; x\in L^{\gamma}, \;\gamma\in
 \Delta.\nonumber
 \end{equation}
Let $V$ be a vector space and $\mathfrak{F}:=\{f_{1},\ldots,f_{n}\}$
a finite set of non-zero linear functions on $V$. Write
$$\Omega_{\mathfrak{F}}:=\{v\in V\mid \Pi_{1\leq i\neq j\leq
n}(f_{i}-f_{j})(v)\neq 0\}.$$  It is a standard fact that
$\Omega_{\mathfrak{F}}\neq\emptyset$  (see also \cite[Lemma
2.2.1]{MBJ}\label{lemma-zarisk}). The following technical lemma is a basic fact in Linear Algebra. For convenience, we write down
a proof:
\begin{lemma}  \label{lemmeigenvector} \textit{Let $\mathfrak{A}$ be an
algebra. For $a\in \mathfrak{A}$ write $L_a$ for the
left-multiplication operator given by $a$. Suppose
$x=x_{1}+x_{2}+\cdots+x_{n}$ is a sum of eigenvectors of $L_a$
associated with mutually distinct eigenvalues. Then all $x_{i}$'s
lie in the subalgebra generated by $a$ and $x$.}
\end{lemma}

 \begin{proof} Let $\lambda_i$ be the eigenvalues of $L_a$ corresponding to
 $x_i$. Suppose for a moment that all the $\lambda_i$'s are nonzero.
 Then
 $$
 (L_a)^{k}(x)=\lambda_{1}^{k}x_{1}+\lambda_{2}^{k}x_{2}+\cdots+\lambda_{n}^{k}x_{n}\quad
 \mbox{for}\; k\geq 1.
 $$
 Our conclusion in this case follows from the fact that the Van der
 Monde
  determinate given by
 $\lambda_1,\lambda_2,\ldots,\lambda_n$ is nonzero and thereby the
 general situation is clear.
 \end{proof}
\begin{lemma}\cite[Proposition 1, p.137]{MS}\label{lem-weight-information}
\begin{itemize}
\item[$\mathrm{(1)}$]
  If $L\neq\mathrm{A}(1,1),$ $\mathrm{P}(3)$  or $\mathrm{Q}(n)$
 then $\mathrm{dim}L^{\gamma}=1$ for every $\gamma\in
\Delta.$
\item[$\mathrm{(2)}$]
If $L\neq {\rm{Q}}(n)$ then $0\notin \Delta_{\bar{1}}.$
\end{itemize}
\end{lemma}

 Notice that, for $L=\mathrm{A}(1,1)$, $\mathrm{P}(3)$ or $\mathrm{Q}(n),$ from Table 1.1  $L_{\bar{0}}$ is a semi-simple Lie
algebra. Then  by a standard result in Lie algebras (see \cite{H})
we have
\begin{eqnarray}
  &&{\rm{dim}}L_{\bar{0}}^{\alpha}=1 \quad  \mbox{for}\;  \alpha\in \Delta_{\bar{0}},\label{eq1739f}
\,\ H=\sum_{\alpha\in
\Delta_{\bar{0}}}[L_{\bar{0}}^{\alpha},L_{\bar{0}}^{-\alpha}].
\end{eqnarray}

\begin{theorem}\label{homogeneous generators}
Any  classical  Lie superalgeba  is generated by 1 element.
\end{theorem}

\begin{proof}  By
Lemma \ref{lem-weight-information}(1),  we treat two cases
separately:\\

  \noindent \textit{Case 1}.  If $L\neq\mathrm{A}(1,1)$, $\mathrm{Q}(n)$ or
$\mathrm{P}(3),$ then all the weight spaces are $1$-dimensional.
Choose any $h\in \Omega_{\Delta_{\bar{1}}}\subset H$ and
  an
 element
$x=\sum_{\gamma\in\Delta_{\bar{1}}}x_{\bar{1}}^{\gamma}$, where
$x_{\bar{1}}^{\gamma}$ is a weight vector of $\gamma$.  By Lemmas
 \ref{lem-weight-information}(2) and  \ref{lemmeigenvector},
  all components $x_{\bar{1}}^{\gamma}$ belong to $\langle x+h\rangle.$  Since
$\mathrm{dim}L^{\gamma}=1,$ we conclude that
$L^{\gamma}\subset\langle x+h\rangle$ for  $\gamma\in
\Delta_{\bar{1}}$  and then  $L_{\bar{1}}\subset L.$ By
\cite[Proposition 1.2.7(1), p.20]{Kac},
$L_{\bar{0}}=[L_{\bar{1}},L_{\bar{1}}]$ and then
$\langle x+h\rangle=L$.\\

  \noindent \textit{Case 2}. Let $L=\mathrm{A}(1,1)$, $\mathrm{Q}(n)$ or
$\mathrm{P}(3).$ In this case, there exists a weight space which is
not 1-dimensional.

Let $L$=A$(1,1)$. For simplicity, write
  $e_{ij}$ for $e_{ij}+\mathbb{F}I_{4}.$  By Table 1.1,
let $L_{\bar{1}}=L_{\bar{1}}^{1}\oplus L_{\bar{1}}^{2}$ be a direct
sum of two irreducible $L_{\bar{0}}$-modules.
  The standard
  basis of A$(1,1)$ is listed below:
$$
\begin{tabular}{|l|l|l|l|}
\multicolumn{4}{c}{ Table 1.2} \\[1pt]
 \hline & $H$  \vline ~~~$e_{11}+e_{33}, \,\
 e_{11}+e_{44}$&&$L_{\bar{1}}^{1}$ \vline~~~$e_{13}, \,\
e_{14}, \,\ e_{23},\,\ e_{24}$\\ \cline{2-2}
\cline{4-4}\raisebox{1ex}[0pt]{$L_{\bar{0}}$}&~~~~~~~$e_{12}, \,\
e_{21}, \,\ e_{34}, \,\ e_{43}$
 &
  \raisebox{1ex}[0pt]{$L_{\bar{1}}$}&$L_{\bar{1}}^{2}$ \vline ~~~$e_{31}, \,\ e_{32}, \,\ e_{41}, \,\ e_{42}$ \\
\hline
\end{tabular}\\
$$
Let $x$ be the sum of all the standard odd  basis elements (weight
vectors) in Table 1.2, that is, $$x=e_{13}+ e_{14}+ e_{23}+
e_{24}+e_{31}+ e_{32}+e_{41}+ e_{42}.$$
 Choose an element $h\in \Omega_{\Delta_{\bar{1}}}.$ Assert $\langle
x+h\rangle =L.$ To that aim, define $\varepsilon_{i}^{a}$ to be the linear
function on $H$  given by
$\varepsilon_{i}^{a}(e_{11}+e_{2+j,2+j})=\delta_{ij}$ for $1\leq
i,j\leq 2.$
 All the odd weights  and the corresponding odd weight vectors are listed below: \\
$$
\begin{tabular}{|c|c|c|c|c|}
\multicolumn{5}{c}{Table 1.3} \\[1pt]
   \hline
   weights & $\varepsilon_{2}^{a}$ & $\varepsilon_{1}^{a}$ & $-\varepsilon_{1}^{a}$ & $-\varepsilon_{2}^{a}$ \\
\hline  vectors & $e_{13}, \,\ e_{42}$ & $e_{14}, \,\ e_{32}$ & $e_{23}, \,\ e_{41}$ & $e_{31}, \,\ e_{24}$\\
   \hline
 \end{tabular}
 \\
$$
Then, by  Lemma \ref{lemmeigenvector} and Table 1.3, the elements
$$e_{13}+e_{42}, \,\ e_{14}+e_{32}, \,\ e_{23}+e_{41}, \,\
e_{31}+e_{24}$$ lie in $\langle x+h \rangle.$ A direct computation
shows that
\begin{eqnarray*}
&&e_{34}=\frac{1}{2}[e_{14}+e_{32},e_{24}+e_{31}], \,\
e_{12}=\frac{1}{2}[e_{14}+e_{32},e_{13}+e_{42}],\\
&&e_{21}=\frac{1}{2}[e_{24}+e_{31},e_{23}+e_{41}], \,\
e_{43}=\frac{1}{2}[e_{13}+e_{42},e_{23}+e_{41}].
\end{eqnarray*}
Then by Table 1.2, all the standard even basis elements $e_{12},
e_{21}, e_{43}, e_{34}\; \mbox{lie in}\; \langle x+h \rangle.$
According to (\ref{eq1739f}), we have $L_{\bar{0}}\subset \langle x+
h\rangle.$ As
$$[e_{14}+e_{32},e_{21}]=-e_{24}+e_{31} \in \langle x+h \rangle,$$
we have $e_{24},  e_{31}\in \langle x+h\rangle.$ Since $e_{24}\in
L_{\bar{1}}^{1},$ $e_{31}\in L_{\bar{1}}^{2}$ (see Table 1.2) and
$L_{\bar{1}}^{i}$ is  irreducible as $L_{\bar{0}}$-module, where
$i=1,2,$  we have $L_{\bar{1}}\subset \langle x+h \rangle.$ Therefore
$\langle
x+h\rangle=L.$\\

 Let $L={\rm{P(3)}}.$ Note that ${\rm{P(3)}}$ is a subalgebra of
A(3,3) consisting of the matrices of the form:
  $\left(
    \begin{array}{c|c}
      a & b\\
      \hline
      c & -a^{\rm{T}} \\
    \end{array}
  \right)$,
  where ${\rm{tr}}(a)=0,$ $b$ is symmetric and $c$ is skew-symmetric
  (see \cite{Kac}). By Table 1.1, let $L_{\bar{1}}=L_{\bar{1}}^{1}\oplus
L_{\bar{1}}^{2}$ be a direct sum of two irreducible
$L_{\bar{0}}$-modules.
The standard  basis of ${\rm{P(3)}}$ is as follows:\\
$$
\begin{tabular}{|l|l|}
\multicolumn{2}{c}{ Table 1.4} \\[1pt]
  \hline &$H$\vline~~~ $~e_{11}-e_{22}-e_{55}+e_{66}$, \; $e_{11}-e_{33}-e_{55}+e_{77}$, \; $e_{11}-e_{44}-e_{55}+e_{88}$\\ \cline{2-2}
  \raisebox{1ex}[0pt]{$L_{\bar{0}}$}&~~~~~~~$~e_{12}-e_{65}, \; e_{13}-e_{75}, \; e_{14}-e_{85}, \; e_{23}-e_{76}, \; e_{24}-e_{86}, \; e_{34}-e_{87},$\\&~~~~~~~$~e_{21}-e_{56}, \;e_{31}-e_{57}, \; e_{41}-e_{58}, \; e_{42}-e_{68}, \; e_{43}-e_{78}, \; e_{32}-e_{67}$\\
 \hline&$ $~\; \;~$ \vline $ ~~~$e_{15},\;  e_{26}, \; e_{18}+e_{45},\;
 e_{28}+e_{46},\;
 e_{38}+e_{47},$\\ \raisebox{1ex}[0pt]{$L_{\bar{1}}$} &\raisebox{1ex}[0pt]{$L_{\bar{1}}^{1}$}~\vline~~~~$e_{37},\;  e_{48},\;  e_{16}+e_{25},\;
 e_{17}+e_{35},\;
 e_{27}+e_{36}$\\
 \cline{2-2}&$L_{\bar{1}}^{2}~$\vline \;~~~$e_{52}-e_{61}, \; e_{53}-e_{71}, \; e_{54}-e_{81}, \; e_{63}-e_{72}, \; e_{64}-e_{82}, \;e_{74}-e_{83}$  \\
  \hline
\end{tabular}
\\
$$
Let $x$ be the sum of all standard odd basis elements (weight
vectors) in Table 1.4. Choose an element $h\in
\Omega_{\Delta_{\bar{1}}}.$ Assert $\langle x+h\rangle=L.$
To that aim, define $\varepsilon_{i}^{p}$ to be the linear function on $H$ given
by
$$\varepsilon_{i}^{p}(e_{11}-e_{1+j,1+j}-e_{55}+e_{5+j,5+j})=\delta_{ij}$$
for $1\leq i,j\leq 3.$
 All the odd weights and the corresponding odd weight vectors are listed below:\\
$$
\begin{tabular}{|c|c|c|c|c|}
\multicolumn{5}{c}{ Table 1.5} \\[1pt]
  \hline
  weight& $2\varepsilon_{1}^{p}+2\varepsilon_{2}^{p}+2\varepsilon_{3}^{p}$&$ -2\varepsilon_{1}^{p}$ & $-2\varepsilon_{2}^{p}$ & $-2\varepsilon_{3}^{p}$  \\
  \hline
  vectors &$ e_{15}$ & $e_{26}$ & $e_{37}$&$ e_{48}$ \\
\hline
  weight & $\varepsilon_{1}^{p}+\varepsilon_{3}^{p}$&$-\varepsilon_{1}^{p}-\varepsilon_{2}^{p}$& $-\varepsilon_{1}^{p}-\varepsilon_{3}^{p}$& \\
  \hline
  vectors & $e_{17}+e_{35},$$e_{64}-e_{82}$&
  $e_{27}+e_{36},$$e_{54}-e_{81}$& $e_{28}+e_{46},$$e_{53}-e_{71}$& \\
  \hline
  weight & $\varepsilon_{2}^{p}+\varepsilon_{3}^{p}$ & $\varepsilon_{1}^{p}+\varepsilon_{2}^{p}$& $-\varepsilon_{2}^{p}-\varepsilon_{3}^{p}$&\\
\hline
 vectors
  &$e_{16}+e_{25},$$e_{74}-e_{83}$& $e_{18}+e_{45},$$e_{63}-e_{72}$&
  $e_{38}+e_{47},$$e_{52}-e_{61}$&\\
  \hline
\end{tabular}
\\
$$
By Lemma \ref{lemmeigenvector} and Table 1.5,  one sees that
$\langle x+h \rangle$ contains the following elements
$$e_{37}, \,\ e_{48},\,\
e_{38}+e_{47}+e_{52}-e_{61},\,\ e_{16}+e_{25}+e_{74}-e_{83},\,\
e_{17}+e_{35}+e_{64}-e_{82},$$
$$e_{15}, \,\ e_{26}, \,\
e_{27}+e_{36}+e_{54}-e_{81},\,\   e_{18}+e_{45}+e_{63}-e_{72}, \,\
e_{28}+e_{46}+e_{53}-e_{71}.$$ Lie superbrackets of the above odd
elements yield
\begin{eqnarray*}
&&e_{57}-e_{31}=[e_{37},e_{28}+e_{46}+e_{53}-e_{71}], \,\
e_{58}-e_{41}=[e_{48},e_{27}+e_{36}+e_{54}-e_{81}],\\
&&e_{56}-e_{21}=[e_{26},e_{38}+e_{47}+e_{52}-e_{61}], \,\
e_{12}-e_{65}=[e_{15},e_{38}+e_{47}+e_{52}-e_{61}],\\
&&e_{13}-e_{75}=[e_{15},e_{28}+e_{46}+e_{53}-e_{71}], \,\
e_{23}-e_{76}=[e_{26},e_{18}+e_{45}+e_{63}-e_{72}],\\
&&e_{24}-e_{86}=[e_{26},e_{17}+e_{35}+e_{64}-e_{82}], \,\
e_{14}-e_{85}=[e_{15},e_{27}+e_{36}+e_{54}-e_{81}],\\
&&e_{34}-e_{87}=[e_{37},e_{16}+e_{25}+e_{74}-e_{83}], \,\
e_{67}-e_{32}=[e_{37},e_{18}+e_{45}+e_{63}-e_{72}],\\
&&e_{68}-e_{42}=[e_{48},e_{17}+e_{35}+e_{64}-e_{82}], \,\
e_{78}-e_{43}=[e_{48},e_{16}+e_{25}+e_{74}-e_{83}].
\end{eqnarray*}
Then, according to Table 1.4 and (\ref{eq1739f}) one sees
$L_{\bar{0}}\subset \langle x+h \rangle.$ Since
$$e_{18}+e_{45}-e_{63}+e_{72}=[e_{17}+e_{35}+e_{64}-e_{82},
e_{78}-e_{43}]\in \langle x+h \rangle,$$ and
$e_{18}+e_{45}+e_{63}-e_{72}\in \langle x+h\rangle,$ we have
$$e_{18}+e_{45},\;e_{63}-e_{72}\in \langle x+h \rangle.$$ By Table
1.4 we have $$e_{18}+e_{45}\in L_{\bar{1}}^{1},\; e_{63}-e_{72}\in
L_{\bar{1}}^{2}.$$ Then the irreducibility of $L_{\bar{1}}^{i}$ as
$L_{\bar{0}}$-module
 ensures that $L_{\bar{1}}^{i}\subset \langle x+h\rangle,$ where $i=1,2.$ Therefore,
 $\langle x+h\rangle=L.$\\

  Let $L= \mathrm{Q}(n).$ Note that ${\rm{Q}}(n)=\widetilde{\rm{Q}}(n)/\mathbb{F}I_{2n+2}$ and
  $\widetilde{\rm{Q}}(n)$ is the subalgebra of $\mathfrak{sl}(n+1,n+1)$ consisting of the matrices of the form:
  $\left(
    \begin{array}{c|c}
      a & b\\
      \hline
      b & a \\
    \end{array}
  \right)$,
  where ${\rm{tr}}(b)=0$ (see \cite{Kac}).
 For simplicity, write $e_{ij}$ for $e_{ij}+\mathbb{F}I_{2n+2}.$  The
standard
  basis of Q($n$) is listed below:\\
$$
\begin{tabular}{|l|l|}
\multicolumn{2}{c}{Table 1.6} \\[1pt]
  \hline &$H$\vline~~~~~ $e_{ii}+e_{n+1+i,n+1+i},\,\ 1\leq i\leq n$\\
 \cline{2-2}  \raisebox{1.6ex}[0pt]{$L_{\bar{0}}$} &~~~~~~~~~$e_{ij}+e_{n+1+i,n+1+j},\,\ 1\leq i\neq j\leq n+1$ \\
 \hline  &~~$e_{1,n+2}+e_{n+2,1}-e_{i,n+1+i}-e_{n+1+i,i},\,\ 2\leq i\leq n+1,$\\\raisebox{1.6ex}[0pt]{$L_{\bar{1}}$}&~~$e_{j,n+1+k}+e_{n+1+j,k}, \,\ 1\leq k\neq j\leq n+1$\\
  \hline
\end{tabular}
\\
$$
Let $x$ be the sum of all standard odd  basis elements (weight
vectors) in Table 1.6. Choose an element $h\in
\Omega_{\Delta_{\bar{1}}}.$ Assert $\langle x+h\rangle=L.$ To that
aim, define $\varepsilon_{i}^{q}$ to be the linear function on $H$
by $\varepsilon_{i}^{q}(e_{jj}+e_{n+1+j,n+1+j})=\delta_{ij}$ for
$1\leq i,j\leq n.$
All the odd weights  and the corresponding odd weight vectors are as follows:\\
$$
\begin{tabular}{|c|c|}
\multicolumn{2}{c}{Table 1.7} \\[1pt]
  \hline
   weights &$\varepsilon_{j}^{q}-\varepsilon_{k}^{q},\,\ 1\leq k\neq j\leq n$\\
  \hline
  vectors &$e_{j,n+1+k}+e_{n+1+j,k}, \,\ 1\leq k\neq j\leq n$ \\\hline
   weights &$\varepsilon_{j}^{q},\,\ 1\leq j\leq n$\\
  \hline
 vectors &$e_{j,2n+2}+e_{n+1+j,n+1}, \,\ 1\leq j\leq n$   \\
  \hline
  weights &$-\varepsilon_{k}^{q}, \,\ 1\leq k\leq n$\\
  \hline
 vectors &$e_{n+1,n+1+k}+e_{2n+2,k}, \,\ 1\leq k\leq n$   \\
 \hline
 weights &$0$\\
  \hline
 vectors &$e_{1,n+2}+e_{n+2,1}-e_{i,n+1+i}-e_{n+1+i,i}, \,\ 2\leq i\leq n+1$\\
 \hline
\end{tabular}\\
$$

 By  Lemma \ref{lemmeigenvector} and Table 1.7,  one sees that $\langle
x+h\rangle$ contains  the following elements
\begin{eqnarray*}
&&e_{j,n+1+k}+e_{n+1+j,k}\; (1\leq k\neq j\leq n),\\
&&e_{j,2n+2}+e_{n+1+j,n+1} \;(1\leq j\leq n),\\
&&e_{n+1,n+1+k}+e_{2n+2,k}\;(1\leq k\leq n),\\
&&\sum_{j=2}^{n+1}(e_{1,n+2}+e_{n+2,1}-e_{j,n+1+j}-e_{n+1+j,j}).
\end{eqnarray*}
Write
 $Z$ for $\sum_{j=2}^{n+1}(e_{1,n+2}+e_{n+2,1}-e_{j,n+1+j}-e_{n+1+j,j}).$
Then
\begin{eqnarray*}
&&[e_{i,n+1+k}+e_{n+1+i,k},Z] \\
=&&\delta_{ij}(e_{jk}+e_{n+1+j,n+1+k})-\delta_{i1}(e_{1k}+e_{n+2,n+1+k})\\&&+\delta_{k1}(e_{i1}+e_{n+1+i,n+2})-\delta_{kj}(e_{ij}+e_{n+1+i,n+1+j}),\;1\leq
i\neq k\leq n+1.
\end{eqnarray*}
Hence
$$e_{ik}+e_{n+1+i,n+1+k}\in\langle x+h\rangle,
 \,\
 1\leq i\neq k\leq n+1.$$
 So, by Table 1.6 and (\ref{eq1739f}) we have $L_{\bar{0}}\subset\langle x+h \rangle.$  Since $x$ lies in
 $\langle x+h \rangle$ and $L_{\bar{1}}$ is irreducible as $L_{\bar{0}}$-module (see Table 1.1), we have
$L_{\bar{1}}\subset\langle x+h\rangle.$ Furthermore, $\langle x+h
\rangle=L.$

\end{proof}

\section{Cartan Lie superalgebras}
All the Cartan Lie superalgebras are listed below \cite{Kac, MS}:
$$
W(n)\,\mbox{with}\, n\geq 3,\,\ S(n)\,\mbox{with}\, n\geq 4, \,\
\widetilde{S}(2m) \, \mbox{with}\, m\geq 2\,\ \mbox{and}\,\ H(n)\,
 \mbox{with}\, n\geq 5.$$
 Let $\Lambda(n)$ be the
Grassmann superalgebra with   generators $\xi_{1},\ldots,\xi_{n}$.
For a $k$-shuffle $u:=(i_{1},i_{2},\ldots,i_{k})$, that is, a
strictly increasing sequence between $1$ and $n$,  we write $|u|:=k$
and $x^{u}:=\xi_{i_{1}}\xi_{i_{2}} \cdots \xi_{i_{k}}.$ Letting
${\rm{deg}}\xi_{i}=1,\; i=1,\ldots,n,$ we obtain the standard
consistent $\mathbb{Z}$-grading of $\Lambda(n).$ Let us briefly
describe the
  Cartan  Lie superalgebras.

  \begin{itemize}\item
 $W(n)={\rm{der}}\Lambda(n)$ is $\mathbb{Z}$-graded,
$W(n)=\oplus_{k=-1}^{n-1}W(n)_{k},$ where
$$W(n)_{k}={\rm{span}}_{\mathbb{F}}\{
x^{u}\partial/\partial\xi_{i}\mid |u|=k+1,\; 1\leq i\leq n\}.$$
\end{itemize}
\begin{itemize}\item
$S(n)=\oplus_{k=-1}^{n-2}S(n)_{k}$ is a $\mathbb{Z}$-graded
subalgebra of $W(n)$, where
$$S(n)_{k}={\rm{span}}_{\mathbb{F}}\{\mathrm{D}_{ij}(x^{u})\mid |u|=k+2,\,\ 1\leq i,j\leq n\}.$$
 Hereafter, ${{\mathrm{D}_{ij}}}(f):=\partial(f)/
\partial\xi_{i}\partial/\partial\xi_{j}+\partial(f)/\partial\xi_{j}\partial/\partial\xi_{i}$
for $f\in \Lambda(n).$
\end{itemize}
\begin{itemize}\item $\widetilde{S}(2m)$ is a subalgebra of $W(2m)$ and
as  a $\mathbb{Z}$-graded subspace,
$$\widetilde{S}(2m)=\oplus_{k=-1}^{2m-2}\widetilde{S}(2m)_{k}\,\ \mbox{
with}\,\ m\geq 2,$$ where
\begin{eqnarray*}&&\widetilde{S}(2m)_{-1}={\rm{span}}_{\mathbb{F}}\{(1+\xi_{1}\cdots\xi_{2m})\partial/\partial\xi_{j}\mid
1\leq j\leq 2m\},\\&&\widetilde{S}(2m)_{k}=S(2m)_{k},\; 0\leq k\leq
2m-2.\end{eqnarray*} Notice that $\widetilde{S}(2m)$ is not a
$\mathbb{Z}$-graded subalgebra of $W(2m)$.
 \end{itemize}
 \begin{itemize}
 \item
 $H(n)=\oplus_{k=-1}^{n-3}H(n)_{k}$ is a
      $\mathbb{Z}$-graded
subalgebra of $W(n)$, where
$$H(n)_{k}={\rm{span}}_{\mathbb{F}}\{{\rm{D_{H}}}(x^{u})\mid|u|=k+2\}.$$
Hereafter, ${\rm{D_{H}}}$ is a linear mapping of $\Lambda(n)$ to
$W(n)$ such that
 ${\rm{D_{H}}}(f):=(-1)^{|f|}\sum_{i=1}^{n}\partial(f)/\partial\xi_{i}\partial/\partial\xi_{i'} $
  for  $f\in \Lambda(n),$ where
 $'$ is the
involution of the index set $\{1,\ldots,n\}$ satisfying that
$i'=i+[\frac{n}{2}]$ for
      $i\leq [\frac{n}{2}]$ and $n'=n$ if $n$ is odd. Here, $[\frac{n}{2}]$ is the biggest integer less than $\frac{n}{2}$ $(n\geq 5).$
 \end{itemize}

In the sequel, we write $W, S, \widetilde{S}, H$ instead of $W(n),
S(n), \widetilde{S}(2m), H(n),$ respectively. Throughout this
section  $L$ denotes one of the Cartan Lie superalgebras $W, S,
\widetilde{S},$ or $ H$. Consider its decomposition of subspaces:
\begin{equation}\label{eqcartangraded}
L=L_{-1}\oplus \cdots \oplus L_{s}.
\end{equation}
For $W, S, \widetilde{S}$ or $H$, the height $s=n-1,$ $n-2,$ $2m-2$
or $n-3,$ respectively. Note that $S $ and $H$ are
$\mathbb{Z}$-graded subalgebras of $W$ with respect to
(\ref{eqcartangraded}), but $\widetilde{S}$ is not. The null $L_{0}$
is isomorphic to $\mathfrak{gl}(n),
\mathfrak{sl}(n),\mathfrak{sl}(2m),\mathfrak{so}(n)$ for $L=W, S,
\widetilde{S}, H,$ respectively.

From \cite{Kac,MS}, we can write down the following facts:
\begin{lemma}\label{lem-Cartan-component} Keep  notations as above.

\begin{itemize}
\item[$\mathrm{(1)}$] The subspace $L_{-1}$
is an irreducible $L_{0}$-module.

\item[$\mathrm{(2)}$] A Cartan Lie superalgebra $L$ is generated by  the local part
$L_{-1}\oplus L_{0}\oplus L_{1}.$

\item[$\mathrm{(3)}$] The subspace $L_{1}$
is an irreducible $L_{0}$-module for $L=S,\widetilde{S}$ or $H,$
except for $H(6).$ For $L=H(6)$ or $W,$ the subspace $L_1$ is a
direct sum of two irreducible $L_0$-submodules.
\end{itemize}
\end{lemma}
The following is a list of bases of the standard Cartan subalgebras
$\mathfrak{h}_{L_0} $ of
$L_{0}.$\\

~~~~~~~~~~~~\begin{tabular}{|l|l|}
\multicolumn{2}{c}{Table 2.1} \\[1pt]\hline
\multicolumn{1}{|c|}{$L$}&\multicolumn{1}{|c|}{$\mathfrak{h}_{L_0}$}\\
\hline
 ~~$W(n)$&
~~~$\xi_{i}\partial/\partial\xi_{i},$ $1\leq
i\leq n$\\
\hline ~~$S(n)$&~~~$
\xi_{1}\partial/\partial\xi_{1}-\xi_{j}\partial/\partial\xi_{j},$
$2\leq j\leq n$\\
\hline ~~$\widetilde{S}(2m)$&~~~$
\xi_{1}\partial/\partial\xi_{1}-\xi_{j}\partial/\partial\xi_{j},$
$2\leq j\leq 2m$\\
\hline
~~$H(2m)$&~~~$\xi_{i}\partial/\partial\xi_{i}-\xi_{i'}\partial/\partial\xi_{i'},$
$1\leq i\leq m$\\
\hline
~~$H(2m+1)$&~~~$\xi_{i}\partial/\partial\xi_{i}-\xi_{i'}\partial/\partial\xi_{i'},$
$1\leq i\leq m$\\
\hline
\end{tabular}\\

The weight space decomposition of the subspace $L_k$ relative to
   $\mathfrak{h}_{L_0}$ is:
$$L_{k}=\delta_{k,0}\mathfrak{h}_{L_0}\oplus_{\alpha\in
\Delta_{k}}L_{k}^{\alpha},\,\ \mbox{where}\,\ -1\leq k\leq s.$$ We
write down the following weight sets which will be used in the proof
of the following Lemma \ref{lem-cartan-weight-information}. For
$W(n)$, define $\varepsilon_{i}^{w}$ to be the linear function on
$\mathfrak{h}_{W_0}$ by
$$\varepsilon_{i}^{w}(\xi_{j}\partial/\partial\xi_{j})=\delta_{ij},\;
1\leq i,j\leq n.$$ We have
\begin{eqnarray}
&&\Delta_{-1}=\{-\varepsilon_{j}^{w}\mid 1\leq j\leq
n\},\nonumber\\
&&\Delta_{1}=\big\{\varepsilon_{k}^{w}+\varepsilon_{l}^{w}-\varepsilon_{j}^{w}\mid
1\leq k\neq l,j\leq n\big\}.\label{eq926}
\end{eqnarray}
For $S(n)$ and $\widetilde{S}(n),$ define $\varepsilon_{i}^{s}$ to
be the linear function on $\mathfrak{h}_{S_{0}}$  by
$$\varepsilon_{i}^{s}(\xi_{1}\partial/\partial\xi_{1}-\xi_{j}\partial/\partial\xi_{j})=\delta_{ij},\;
2\leq i,j\leq n$$ and write
$\varepsilon_{1}^{s}:=\sum_{l=2}^{n}\varepsilon_{l}^{s}.$ We have
\begin{eqnarray*}
&&\Delta_{-1}=\{\varepsilon_{j}^{s}\mid 1\leq j\leq
n\},\\
&&\Delta_{1}=\big\{\varepsilon_{k}^{s}+\varepsilon_{l}^{s}-\varepsilon_{j}^{s}\mid
1\leq k\neq l,j\leq n\big\}.
\end{eqnarray*} For $H(2m)$, define $\varepsilon_{i}^{h}$ to be the linear function on
$\mathfrak{h}_{H_{0}}$  by
$$\varepsilon_{i}^{h}(\xi_{j}\partial/\partial\xi_{j}-\xi_{j'}\partial/\partial\xi_{j'})=\delta_{ij},\;
 1\leq i,j\leq m.$$ We have
\begin{eqnarray}
&&\Delta_{-1}=\{\ \pm\varepsilon_{j}^{h}\mid 1\leq j\leq
m\},\nonumber\\
&&\Delta_{1}=\{\pm(\varepsilon_{i}^{h}+\varepsilon_{j}^{h})\pm\varepsilon_{k}^{h},
\pm(\varepsilon_{i}^{h}-\varepsilon_{j}^{h})\pm\varepsilon_{k}^{h}\mid
1\leq i< j< k\leq
m\}\nonumber\\&&~~~~~~~~\cup\{\pm\varepsilon_{l}^{h}\mid 1\leq l\leq
m\}.\label{eqheven23}
\end{eqnarray}
For $H(2m+1)$, define $\varepsilon_{i}^{h'}$ to be the linear
function on $\mathfrak{h}_{H_{0}}$ by
$$\varepsilon_{i}^{h'}(\xi_{j}\partial/\partial\xi_{j}-\xi_{j'}\partial/\partial\xi_{j'})=\delta_{ij},\;
 1\leq i,j\leq m.$$  We have
\begin{eqnarray*}
&&\Delta_{-1}=\{0\}\cup \{\pm\varepsilon_{i}^{h'}\mid 1\leq i\leq m\},\\
&&\Delta_{1}=\{0\}\cup\{\pm\varepsilon_{l}^{h'},\pm(\varepsilon_{i}^{h'}+\varepsilon_{j}^{h'}),\pm(\varepsilon_{i}^{h'}-\varepsilon_{j}^{h'})\mid
1\leq l\leq m, 1\leq i< j\leq m\}\\
&&~~~~~~~~\cup\{\pm(\varepsilon_{i}^{h'}+\varepsilon_{j}^{h'})\pm\varepsilon_{k}^{h'},
\pm(\varepsilon_{i}^{h'}-\varepsilon_{j}^{h'})\pm\varepsilon_{k}^{h'})\mid
1\leq i< j< k\leq m\}.\end{eqnarray*}

  For
$L_{1}=H(6)_{1}$ or $W_{1},$ then by Lemma
\ref{lem-Cartan-component}(3),
 $L_{1}$ is a direct sum of two irreducible $L_{0}$-modules
$$L_{1}=L_{1}^{1}\oplus L_{1}^{2}.$$ Let $\Delta_{1}^{i}$
be the weight set of $L_{1}^{i},$ $i=1,2.$
\begin{lemma}\label{lem-cartan-weight-information}With the above
nonations, we have the following properties: ~
\begin{itemize}
\item[$\mathrm{(1)}$]
If $L=W$ then $\Delta_{-1}\cap \Delta_{1} =\emptyset.$
\item[$\mathrm{(2)}$]
If  $L=S$ or $\widetilde{S}$   then $\Delta_{-1}\cap
\Delta_{1}=\emptyset.$
\item[$\mathrm{(3)}$]
If  $L=H$  then  $\Delta_{-1}\neq \Delta_{1}.$
\item[$\mathrm{(4)}$] If $L=H(6)$ or $W,$
 there exist nonzero weights $\alpha_{1}^{i}\in \Delta_{1}^{i} $
such that $\alpha_{1}^{1}\neq \alpha_{1}^{2}.$
\end{itemize}
\end{lemma}
\begin{proof}
All the statements follow directly from the above computations
except (4) for $L=H(6)$
  or $W.$
  In this case, from (\ref{eq926}) and (\ref{eqheven23})  one sees that $0\notin \Delta_{1} $ and $|\Delta_{1}|>1$.
  Consequently,  (4) holds.
\end{proof}
For $L=W,S,\widetilde{S}$ or $H,$  fix the corresponding standard
Cartan subalgebra, respectively.
\begin{lemma}\label{zero}
 For $L,$ there exists a weight vector $ x_{1}\in
L_{1}^{\alpha}$ for some $\alpha\in \Delta_{1}$ such that
$[x_{1},x_{1}]=0.$
\end{lemma}
\begin{proof}
It is easy to see that a standard basis vector of $L_{1}$ is also a
weight vector for some weight $\alpha\in \Delta_{1}.$ For $L=W,$ we
have
$$[x_{j}x_{k}\partial/\partial\xi_{i},x_{j}x_{k}\partial/\partial\xi_{i}]=0$$
where $1\leq i,j,k\leq n $ and
$i,j,k$ are pairwise distinct. For $L=S$ or $\widetilde{S},$ we have
$$[\mathrm{D}_{ij}(x_{i}x_{j}x_{k}),\mathrm{D}_{ij}(x_{i}x_{j}x_{k})]=0$$
where $1\leq i,j,k\leq n $ and $i,j,k$ are pairwise distinct. For
$L=H$, we have
$$[\mathrm{D}_{{\rm{H}}}(x_{i}x_{j}x_{k}),\mathrm{D}_{\rm{H}}(x_{i}x_{j}x_{k})]=0$$
where $1\leq i,j,k\leq [\frac{n}{2}]$ and $ i,j,k $ are pairwise
distinct.
\end{proof}

Recall that the null $L_{0}$ is isomorphic to $\mathfrak{gl}(n),
\mathfrak{sl}(n),\mathfrak{sl}(2m),\mathfrak{so}(n)$ for $L=W, S,
\widetilde{S}$ or $H,$ respectively.

    Let $\mathfrak{g}$ be a simple Lie algebra. If $\mathfrak{h}$ is a Cartan subalgebra of
$\mathfrak{g}$,
 $x\in\mathfrak{g}$ is called $\mathfrak{h}$-\textit{balanced}
 provided that  $x^{\alpha}\neq0$ for  $\alpha\in \Phi,$ where $\Phi\subset \mathfrak{h}^{*}$ is the root
 system of $\mathfrak{g}$ relative to $\mathfrak{h}.$ The lemma below is abstracted from
 \cite{MBJ}:

\begin{lemma}
If $x$ is a non-zero element in $\mathfrak{g}$ then there exists
some Cartan subalgebra $\mathfrak{h}$  such that $x$ is
  $\mathfrak{h}$-balanced.
  \end{lemma}

\begin{proof}  Suppose $x$
is  a non-zero element of $\mathfrak{g}$ and let $\mathfrak{h}'$ be
a Cartan subalgebra of $\mathfrak{g}$. By the proof of \cite[Theorem
2.2.3]{MBJ}, there exists $\varphi\in \mathrm{Aut}\mathfrak{g}$ such
that $\varphi(x)$ is $\mathfrak{h}'$-balanced. Letting
$\mathfrak{h}=\varphi^{-1}(\mathfrak{h}')$, one sees that
$\mathfrak{h}$ is a Cartan subalgebra and $x$ is
$\mathfrak{h}$-balanced.
\end{proof}
 From
 \cite{MBJ}, we write down two useful facts:

  $\bullet$  If $x\in\mathfrak{g}$  is
  $\mathfrak{h}$-balanced,  then
  $\mathfrak{g}=\langle x, h\rangle$ for  $h\in
 \Omega_{\Phi}$, where $\Phi\subset \mathfrak{h}^{*}$ is the root
 system of $\mathfrak{g}$ relative to $\mathfrak{h}.$

 $\bullet$ If $\mathfrak{sl}(n)$ is generated
by an $\mathfrak{h}$-balanced element $x$ and an
  element $h$ in $\Omega_{\Phi}$, then
$\mathfrak{gl}(n)$ is generated by $x$ and $h+z$,   where $z$ is a
nonzero central element in $\mathfrak{gl}(n)$.

Note that the $\mathbb{Z}$-grading of a Cartan Lie superalgebra is
consistent with the $\mathbb{Z}_{2}$-grading  over $\mathbb{F}.$

\begin{theorem}\label{homogeneous generators cartan}
Any  Cartan  Lie superalgeba  is generated by 1 element.
\end{theorem}

\begin{proof}
For $L=S,$ $\widetilde{S}$ or $H,$ except $H(6),$  fix  the standard
Cartan subalgebra $\mathfrak{h}.$ By  Lemmas
\ref{lem-cartan-weight-information} and \ref{zero} we choose
$\alpha_{-1}\neq \alpha_{1}$ for $\alpha_{i}\in \Delta_{i},\,\
i=-1,1$ and $x_{-1}\in L_{-1}^{\alpha_{-1}}$ and $x_{1}\in
L_{1}^{\alpha_{1}}$ such that $[x_{-1},x_{1}]\neq 0$ and
$[x_{1},x_{1}]=0.$ Let $x_{0}:=2[x_{-1},x_{1}].$ Choose a suitable
Cartan subalgebra $\mathfrak{h'}$ such that $x_{0}$ is an
$\mathfrak{h'}$-balanced element of $L_{0}.$ Let
$$x:=x_{-1}+x_{0}+h'+x_{1}$$ for $h'\in \Omega_{\Delta'_{0}},$ where
$\Delta'_{0}\subset \mathfrak{h}^{'*}$ is the root
 system of $L$ relative to $\mathfrak{h'}.$ Then
we have $$x_{-1}+x_{1},\,\ x_{0}+h'\in \langle x\rangle$$  and then
$$x_{0}=[x_{-1}+x_{1},x_{-1}+x_{1}]=2[x_{-1},x_{1}]\in \langle
x\rangle.$$   Furthermore, we have $h'\in \langle x\rangle.$ Then
$L_{0}=\langle x_{0}, h'\rangle\subset \langle x\rangle.$ Choose an
element
$$h\in\Omega_{\{\alpha_{-1},\alpha_{1}\}}\subset \mathfrak{h}\subset
L_{0}.$$ Then, by Lemma \ref{lemmeigenvector} we obtain that
$x_{-1}$ and $x_{1}$ lie in $\langle x\rangle.$ According to Lemma
\ref{lem-Cartan-component}(1) and (3), the irreducibility of
$L_{-1}$ and $L_{1}$ as $L_{0}$-modules ensures that
$L_{i}\subset\langle x\rangle, i=-1,1.$ Furthermore, $L=\langle
x\rangle.$

For $L=H(6)$ or $W,$ fix  the standard Cartan subalgebra
$\mathfrak{h}.$ From Lemmas \ref{lem-cartan-weight-information} and
\ref{zero} we choose $\alpha_{-1}, \alpha_{1}^{1}$ and
$\alpha_{1}^{2}$ are pairwise distinct for $\alpha_{-1}\in
\Delta_{-1},$ $\alpha_{1}^{i}\in \Delta_{1}^{i}, i=1,2$ and
$x_{-1}\in L_{-1}^{\alpha_{-1}}$ and $x_{1}^{i}\in
L_{1}^{\alpha_{1}^{i}}$ for $i=1,2$ such that
$[x_{-1},x_{1}^{1}+x_{1}^{2}]\neq 0$ and
$[x_{1}^{1}+x_{1}^{2},x_{1}^{1}+x_{1}^{2}]=0.$
 Let
$x_{0}:=2[x_{-1},x_{1}^{1}+x_{1}^{2}].$ Choose a suitable Cartan
subalgebra $\mathfrak{h'}$ such that $x_{0}$ is an
$\mathfrak{h'}$-balanced element of $L_{0}.$ Let
$$x:=x_{-1}+x_{0}+\delta_{L,W}z+h'+x_{1}^{1}+x_{1}^{2}$$ for $0\neq
z\in C(W_{0})\; \mbox{and}\; h'\in \Omega_{\Delta'_{0}},$ where
$\Delta'_{0}\subset \mathfrak{h}^{'*}$ is the root
 system of $L$ relative to $\mathfrak{h'}.$ Then we have
 $$x_{-1}+x_{1}^{1}+x_{1}^{2},\,\
x_{0}+\delta_{L,W}z+h'\in \langle x\rangle$$ and then
$$x_{0}=[x_{-1}+x_{1}^{1}+x_{1}^{2},x_{-1}+x_{1}^{1}+x_{1}^{2}]=2[x_{-1},x_{1}^{1}+x_{1}^{2}]\in
\langle x\rangle.$$   Furthermore, we have $h'+\delta_{L,W}z\in
\langle x\rangle.$ Then $$L_{0}=\langle x_{0}, h'+\delta_{L,W}z
\rangle\subset\langle x\rangle.$$ Choose an element
$$h\in\Omega_{\{\alpha_{-1},\alpha_{1}^{1},\alpha_{1}^{2}\}}\subset
\mathfrak{h}\subset L_{0}.$$ Then, by Lemma \ref{lemmeigenvector} we
obtain that $x_{-1},x_{1}^{1}$ and $x_{1}^{2}$ lie in $\langle
x\rangle.$ According to Lemma \ref{lem-Cartan-component}(1) and (3),
the irreducibility of $L_{-1},$ $L_{1}^{1}$ and $L_{1}^{2}$ as
$L_{0}$-modules ensures that $L=\langle x\rangle.$
\end{proof}

Theorems \ref{homogeneous generators} and \ref{homogeneous
generators cartan} combine to the main result of this paper:
\begin{theorem}
Any simple Lie superalgebra is generated by 1 element.
\end{theorem}

\end{document}